\documentclass[journal]{IEEEtran}
\usepackage[T1]{fontenc}
\usepackage{amsmath}
\usepackage{subfigure}
\usepackage{graphicx}
\usepackage{url}
\usepackage{algorithm}
\usepackage{algorithmic}
\usepackage{amssymb}
\usepackage{color}
\usepackage{multimedia}
\usepackage{setspace}
\usepackage{amsthm}
\usepackage{multirow}
\DeclareMathOperator*{\argmin}{argmin}

\usepackage{pdfsync}

\renewcommand{\Re}{{\mathbb{R}}}
\usepackage{amssymb}

 \def\xx{{\boldsymbol{x}}}
 \def\B{{B}}

\def\X{{\boldsymbol{X}}}  
 \def\BB{{\bf B}}
\def\AA{{\boldsymbol{A}}}  
  
\def\Y{{\boldsymbol{Y}}}  
\def\PPhi{{\boldsymbol{\Phi}}}  
\def\zz{{\boldsymbol{z}}} 
 \def\yy{{\boldsymbol{y}}}
\def\bb{{\boldsymbol{b}}}  
\def\aa{{\boldsymbol{a}}}  
\def\cc{{\boldsymbol{c}}}  
\def\aa{{\boldsymbol{a}}}  
\def\Q{{\boldsymbol{Q}}}  
 \DeclareMathOperator*{\rank}{rank}
  
\def\subjto{{\mbox{subj. to}}}

\newcommand{\Ce}{\mathbb{C}} 
 \newtheorem{thm}{Theorem}
\usepackage{authblk}
\DeclareMathOperator*{\soft}{soft}

\RequirePackage{xspace} \RequirePackage{amsmath, amssymb,color}

\DeclareMathOperator*{\Tr}{Tr}

 \def\xx{{\boldsymbol{x}}}
 \def\B{{B}}
 \def\BB{{\bf B}}
 \def\zz{{\boldsymbol{z}}} 
 \def\yy{{\boldsymbol{y}}}
\def\bb{{\boldsymbol{b}}}  
\def\aa{{\boldsymbol{a}}}  

\def\subjto{{\mbox{subj. to}}}

\renewcommand{\Re}{{\mathbb{R}}}
\newcommand{\T}{\mathsf{T}}
\newcommand{\eg}{\textit{e.g.,~}}
\newcommand{\ie}{\textit{i.e.,~}}
\renewcommand{\H}{\mathsf{H}}

\newtheorem{df}[thm]{Definition}
\newtheorem{cor}[thm]{Corollary}

\def\Z{{\boldsymbol{Z}}} 
\def\I{{\boldsymbol{I}}} 
 \def\Bt{\tilde{B}}
 \def\bigO{\mathcal{O}}

\author{{Henrik~Ohlsson},~\IEEEmembership{Member,~IEEE},
Allen~Y.~Yang,~\IEEEmembership{Member,~IEEE},
{Roy~Dong},
{Michel Verhaegen},
S.~Shankar~Sastry,~\IEEEmembership{Fellow,~IEEE}
\thanks{ Ohlsson, Yang, Dong, and Sastry are with the Department of Electrical Engineering and Computer
  Sciences, University of California, Berkeley, CA, USA. Ohlsson is also with the
Division of Automatic Control, Department of Electrical Engineering, Link\"oping University, Sweden.
Verhaegen is with the Delft Center for Systems and Control, Delft University, Delft 2628CD, The Netherlands.
Corresponding author: Henrik Ohlsson, Cory Hall, University of California, Berkeley, CA 94720. Email: ohlsson@eecs.berkeley.edu.}                        
\thanks{The authors gratefully acknowledge support by the Swedish Research
  Council in the Linnaeus center CADICS, the European Research Council
   under the advanced grant LEARN, contract 267381, by a postdoctoral grant from the Sweden-America
   Foundation, donated by ASEA's Fellowship Fund, by a postdoctoral
   grant from the Swedish Research Council, and by ARO 63092-MA-II and
   DARPA FA8650-11-1-7153.}
\thanks{This paper was presented in part at NIPS 2012, Lake Tahoe,
  USA, Dec 3-6, 2012, \cite{Ohlsson:12}.}
}

\begin{document}

\title{Quadratic Basis Pursuit}

\maketitle

\begin{abstract}                         
In many compressive sensing problems today, the
relationship between the measurements and the unknowns could be
nonlinear. Traditional treatment of such nonlinear relationships have
been to approximate the nonlinearity via a linear model and the
subsequent un-modeled dynamics as noise. The ability to more
accurately characterize nonlinear models has the potential to improve
the results in both existing compressive sensing applications and
those where a linear approximation does not suffice, \eg phase
retrieval. In this paper, we extend the classical compressive sensing
framework to a second-order Taylor expansion of the
nonlinearity. Using a lifting technique and a method we call quadratic
basis pursuit, we show that the sparse
signal can be recovered exactly when the sampling rate is sufficiently
high. We further present efficient numerical algorithms to recover sparse signals in second-order nonlinear systems, which are considerably more difficult to solve than their linear counterparts in sparse optimization.
\end{abstract}


\section{Introduction}
\label{sec:introduction}
Consider the problem of finding the sparsest signal $\xx$ satisfying a system of linear equations:
\begin{equation}\label{eq:l0prob}
\begin{aligned}
\min_{\xx \in \Re^n} &\quad \|\xx\|_0\\
\subjto& \quad y_i =\bb_i^\T \xx,\quad y_i\in \Re,\, \bb_i \in \Re^n  ,\, i=1,\dots,N.
\end{aligned}\end{equation}
This problem is known to be combinatorial and NP-hard
\cite{Natarajan:95} and a number of approaches
to approximate its solution have been proposed. 
One of the most well known approaches is to relax the zero norm
and replace it with the $\ell_1$-norm:
\begin{equation}\label{eq:BP}
\min_{\xx \in \Re^n} \|\xx\|_{1} \quad \subjto\quad y_i= \bb_i^\T \xx,\quad i=1,\dots,N.
\end{equation}
This approach is often referred to as \textit{basis
pursuit} (BP) \cite{Chen:98}.

The ability to recover the optimal solution to \eqref{eq:l0prob} is essential 
in the theory of \textit{compressive sensing} (CS)
\cite{Candes:06,Donoho:06} and a tremendous 
amount of work has been dedicated to solving and analyzing the solution
of \eqref{eq:l0prob} and \eqref{eq:BP} in the last decade. Today CS is
regarded as a powerful tool in signal processing and widely used in
many applications. 
For a detailed review of the literature, the reader is referred to several recent publications such as \cite{bruckstein:09,Eldar:2012}.

It has recently been shown that 
CS can be extended to nonlinear models. More specifically, the relatively new
topic of \emph{nonlinear compressive sensing} (NLCS) deals with a
more general problem of finding the sparsest signal $\xx$ to a
nonlinear set of equations:
\begin{equation}
\begin{aligned}
\min_{\xx  \in \Re^n} & \quad \|\xx\|_0\\
\subjto & \quad  y_i =f_i(\xx),\quad y_i\in \Re,\, i=1,\dots,N,
\end{aligned}\label{eq:nonlinear}
\end{equation}
where each $f_i:\Re^n \rightarrow \Re$ is a continuously
differentiable function. Compared to CS, the literature on  NLCS  is still very limited. The interested reader is referred to \cite{Beck:2012,BlumensathT:2012} and references therein.

In this paper, we will restrict our attention from rather general
nonlinear systems, and instead focus on nonlinearities that depends
quadratically on the unknown $\xx$. More specifically, we consider the following problem formulated in the complex domain:
\begin{equation}
\begin{aligned}[rcl]
\min_{\xx \in \Ce^n} & \quad\| \xx\|_0\\
 \subjto & \quad  y_i =a_i+\bb_i^\H \xx+ \xx^\H\cc_i+\xx^\H \Q_i
\xx,\\ &\quad i=1,\dots,N, 
\end{aligned}
\label{eq:costraint}
\end{equation}
where $a_i\in \Ce$, $\bb_i, \cc_i\in \Ce^n$, $y_i\in \Ce$, and $\Q_i \in \Ce^{n\times n},\,
i=1,\dots,N$. 
In a sense, being able to solve \eqref{eq:costraint} would make it possible to apply the principles of CS to a
second-order Taylor expansion of the nonlinear relationship in
\eqref{eq:nonlinear}, while traditional CS mainly considers its linear
approximation or first-order Taylor expansion. In particular, in the most simple
case, when a second order  Taylor expansion  is taken around zero (\ie a Maclaurin
expansion), let $a_i=f_i(0)$, $\bb_i=\cc_i=\nabla_{\xx}^\T f_i(0)/2$ and
$\Q_i=\nabla^2_{\xx} f_i(0)/2$, $i=1,\dots,N$, with $\nabla_\xx$ and $\nabla^2_\xx$
denoting the gradient and Hessian with respect to $\xx$. In this case, $\Q$ is a Hermitian matrix. Nevertheless, we note that our derivations in the
paper does not depend on the matrix $\Q$ to be symmetric in the real domain or Hermitian in the complex domain.

In another motivating example, we  consider the well-known phase retrieval problem in 
x-ray crystallography, see for instance \cite{KohlerD1972,GonsalvesR1976,GerchbergR1972,FienupJ1982,MarchesiniS2007,BalanR2006}. The underlying principal of x-ray
crystallography is that the information about the crystal structure can be obtained from its diffraction pattern by hitting the crystal by an x-ray beam.
Due to physical limitations, typically only the intensity of the diffraction pattern can be measured but not its phase. This leads to
a nonlinear relation
\begin{equation}\label{eq:PR}
y_i=| \aa_i^\H \xx|^2=\xx^{\H} \aa_i \aa_i^\H \xx,\quad i=1,\dots,N,
\end{equation}
between the measurements $y_1,\dots,y_N \in \Re$ and the structural information
contained in $\xx\in \Ce^n$. The complex vectors $\aa_1, \dots, \aa_N \in \Ce^{n }$
are known and $\H$ denotes the conjugate transpose. The mathematical problem of recovering $\xx$ from
$y_1,\dots,y_N,$ and  $\aa_1, \dots, \aa_N$ is referred to as the
phase retrieval problem. The traditional phase retrieval problem is
known to be combinatorial \cite{Candes:11}. 

If $\xx$ is sparse under an appropriate basis in \eqref{eq:PR}, the problem is referred to as \textit{compressive phase retrieval} (CPR) in \cite{MoravecM2007,ohlsson:11m} or \textit{quadratic compressed sensing} (QCS) in \cite{Shechtman:11}. These algorithms can be applied to several important imaging applications, such as
diffraction imaging \cite{Bunk:07},
astronomy \cite{Dainty:87,Fienup:93}, optics \cite{WaltherA1963},
x-ray tomography \cite{Dierolf:10}, microscopy \cite{miao:08,Antonello:12,Szameit:12},
and quantum mechanics \cite{Corbett:06}, to mention a few. 
As we will later show, our solution as a convex relaxation of \eqref{eq:costraint}, called \emph{quadratic basis pursuit} (QBP), can be readily applied to solve this type of problems, namely, let $a_i=\bb_i=\cc_i=0,\,\Q_i=\aa_i \aa_i^\H,\, i=1,\dots,N$.


\subsection{Contributions}

The main contribution of this paper is a novel convex technique for
solving the sparse quadratic problem \eqref{eq:costraint}, namely, QBP. 
The proposed framework is not a greedy algorithm and inherits desirable properties, \eg 
perfect recovery, from BP and the traditional CS results. In
comparison, most of the existing solutions for sparse nonlinear problems
are greedy algorithms, and therefore their ability to give global convergence guarantees is limited.

Another contribution is an efficient numerical algorithm that solves the QBP problem and compares favorably to other existing sparse solvers in convex optimization.
The algorithm is based on \emph{alternating direction method of
  multipliers} (ADMM). Applying the algorithm to the complex CPR
problem, we show that the QBP approach achieves the state-of-the-art result compared to other phase retrieval solutions when the measurements are under-sampled.

In Section \ref{sec:QCS}, we will first develop the main theory of
QBP. In Section \ref{sec:algorithms}, we present the ADMM algorithm. Finally, in Section \ref{sec:experiments}, we conduct comprehensive experiments to validate the performance of the new algorithm on both synthetic and more practical imaging data.

\subsection{Literature Review}
To the best of our knowledge, this paper is the first work focusing on
recovery of sparse signals from systems of general quadratic
equations. Overall, the literature on nonlinear sparse problems and
NLCS is also very limited. One of the first papers discussing these
topics is \cite{eps151911}. They present a greedy gradient based
algorithm for estimating the sparsest solution to a general nonlinear
equation system.  A greedy approach was also proposed in \cite{Li09}
for the estimation of sparse solutions of nonlinear equation systems. The work of \cite{Beck:2012} proposed several iterative hard-thresholding and sparse simplex pursuit algorithms. As the algorithms are nonconvex greedy solutions, the analysis of the theoretical convergence only concerns about their local behavior. In \cite{BlumensathT:2012}, the author also considered a generalization of the \textit{restricted isometry property} (RIP) to support the use of similar iterative hard-thresholding algorithms for solving general NLCS problems.

Our paper is inspired by several recent works on CS applied to the phase retrieval
problem
\cite{MoravecM2007,Marchesin:09,Chai:10,Shechtman:11,ohlsson:11m,Osherovich:12,Szameit:12,Jaganathan12,schniter12,Shechtman13}. 
In particular, the generalization of compressive sensing
to CS was first proposed in \cite{MoravecM2007}. In \cite{Shechtman:11}, the problem was also referred to as QCS.
These methods typically do not consider a general quadratic constraint as in
\eqref{eq:costraint} but a pure quadratic form
(\ie $a_i=\bb_i=\cc_i=0,\,i=1,\dots,N$, in \eqref{eq:costraint}).

In terms of the numerical algorithms that solves the CPR problem, most of the existing methods are 
greedy algorithms, where a solution to the underlying non-convex
problem is sought by a sequence of local decisions
\cite{MoravecM2007,Marchesin:09,Shechtman:11,Osherovich:12,Szameit:12,Shechtman13}. In
particular, the QCS algorithm in \cite{Shechtman:11} used a \emph{lifting} technique similar to that in
\cite{shor87,Lovász91,Nesterov98,Goemans:1995} and \emph{iterative
  rank minimization} resulting in a series of semidefinite programs (SDPs) that would converge to a local optimum.

The first work that applied the lifting technique to the PR and CPR problems was presented in \cite{Chai:10}.
Extensions of similar techniques were also studied in
\cite{Li:2012,Jaganathan12}.
The methods presented in our previous publications \cite{Ohlsson:12,ohlsson:11m} were also based on the lifting technique.
It is important to note that the algorithms proposed in \cite{Chai:10,Ohlsson:12,ohlsson:11m} are non-greedy global solutions, which are different from the previous local solutions \cite{MoravecM2007,Shechtman:11}. 
Our work was inspired by the solutions to phase retrieval via low-rank approximation in \cite{Chai:10,Candes:11,Candes:11b}. 
Given an
oversampled phase retrieval problem, a
lifting technique was used to relax the nonconvex 
problem with a SDP. 
The authors of \cite{Candes:11,Candes:11b} also derived an upper-bound for the sampling
 rate that guarantees exact recovery in the noise-free case and stable
 recovery in the noisy case. Nevertheless, the work in
\cite{Candes:11,Candes:11b} only addressed the oversampled phase
retrieval problem but not CPR or NLCS. The only similarities
between our work and theirs are the
lifting technique and convex relaxation. This lifting technique has also been used in
other topics to convert nonconvex quadratic problems to SDPs, see for
instance \cite{Waldspurger12,Jaganathan12}. The work presented in
\cite{Chai:10} and our previous contributions
\cite{Ohlsson:12,ohlsson:11m} only discussed the CPR problem.

Finally, in \cite{schniter12}, a \emph{message passing}
algorithm similar to that in CS was proposed to solve the compressive
phase retrieval problem. The work in \cite{Eldar:12} further
considered stability and uniqueness in real phase retrieval problems.
CPR has also been shown useful in practice and we refer the interested
reader to \cite{MoravecM2007,Szameit:12} for two very nice
contributions. Especially fascinating we find the work presented in
\cite{Szameit:12}  where the authors show how CPR can be used to facilitate
sub-wavelength imaging  in microscopy.

\subsection{Notation and Assumptions}
In this paper, we will use bold face to denote vectors and matrices and normal font
for scalars. 
We denote
the transpose of a real vector by  $\xx^\T$ and the conjugate
transpose of a complex vector by $\xx^\H$.  $\X_{i,j}$ is used to denote the
$(i,j)$th element, $\X_{i,:}$ the $i$th row and  $\X_{:,j}$ the $j$th
column of a matrix $\X$, respectively. We will use the notation
$\X_{i_1:i_2,j_1:j_2}$ to denote a submatrix constructed from
rows $i_1$ to $i_2$ and columns  $j_1$ to $j_2$ of $\X$. Given two matrices $\X$ and $\Y$,
we use the following fact that their product in the trace function commutes, namely, $\Tr(\X \Y) = \Tr(\Y \X)$, under the assumption
that the dimensions match. $\| \cdot \|_0$ counts the number of nonzero elements in a vector or matrix; similarly, \linebreak $\| \cdot\|_1$ denotes the element-wise $\ell_1$-norm of a vector or matrix, \ie, the sum of the magnitudes of the elements; whereas $\| \cdot \|$ represents the $\ell_2$-norm for vectors and the spectral norm for matrices.

\section{Quadratic Basis Pursuit}
\label{sec:QCS}

\subsection{Convex Relaxation via Lifting}

As optimizing the $\ell_0$-norm function in \eqref{eq:costraint} is known to be a combinatorial problem, in this section, we first introduce a convex relaxation of  \eqref{eq:costraint}. 

It is easy  to see that the general quadratic constraint of \eqref{eq:costraint} can be rewritten as the quadratic form:
\begin{equation}
y_i  =\begin{bmatrix} 1 &\xx^\H  \end{bmatrix} \begin{bmatrix} a_i &
  \bb_i^\H \\ \cc_i&
  \Q_i \end{bmatrix} \begin{bmatrix} 1\\ \xx \end{bmatrix} \in \Ce, \quad i=1,\dots,N.
\end{equation}
Since each $y_i$ is a scalar, we further have
\begin{align}
y_i  =&\Tr \left( \begin{bmatrix} 1 &\xx^\H  \end{bmatrix} \begin{bmatrix} a_i & \bb_i^\H \\ \cc_i&
  \Q_i \end{bmatrix} \begin{bmatrix} 1\\ \xx \end{bmatrix} \right)
\\=&\Tr \left(\begin{bmatrix} a_i & \bb_i^\H \\ \cc_i &
  \Q_i \end{bmatrix} \begin{bmatrix} 1\\ \xx \end{bmatrix} \begin{bmatrix} 1 &\xx^\H  \end{bmatrix} \right).
\end{align}
Define $\PPhi_i=\begin{bmatrix} a_i & \bb_i^\H \\ \cc_i & \Q_i \end{bmatrix} $ 
and $\X=\begin{bmatrix} 1\\ \xx \end{bmatrix} \begin{bmatrix} 1 &
  \xx^\H \end{bmatrix} $,  both matrices  of dimensions $(n+1) \times
(n+1)$. The operation that constructs $\X$ from the vector
$\begin{bmatrix} 1\\ \xx \end{bmatrix}$ is known as the \emph{lifting}
operator \cite{shor87,Lovász91,Nesterov98,Goemans:1995}. By
definition, $\X$ is a Hermitian matrix, and it satisfies the constraints that $\X_{1,1}=1$ and $\text{rank}(\X)=1$. Hence, \eqref{eq:costraint} can be rewritten as
\begin{equation}\label{eq:firstcond}
\begin{array}{rl}
\min_{\X} & \|\X\|_0 \\
\subjto\quad &y_i =\Tr( \PPhi_i \X),\quad i=1,\dots,N,\\
&\text{rank}(\X)=1, \X_{1,1}=1, \, \X \succeq 0.
\end{array}
\end{equation}
When the optimal solution $\X^*$ is found, the unknown $\xx$ can
be obtained by the rank-1 decomposition of $\X^*$ via \emph{singular
value decomposition} (SVD). 

The above problem is still non-convex and combinatorial. Therefore, solving it for any
moderate size of $n$ is impractical. Inspired by recent
literature on matrix completion \cite{CandesR:08,Chai:10,Candes:11,Candes:11b} and sparse PCA \cite{Aspremont:07}, we relax the
problem into the following convex \emph{semidefinite program} (SDP):
\begin{equation}\label{eq:noiseless-SDP}
\begin{array}{rl}
\min_{\X} &\Tr(\X)+\lambda \|\X\|_1 \\
\subjto \quad &y_i =\Tr ( \Phi_i \X ),\quad i=1,\dots,N,\\
&\X_{1,1}=1,\, \X \succeq 0,
\end{array}
\end{equation}
where $\lambda \geq 0$ is a design parameter. In particular, the trace
of $\X$ is a convex surrogate of the low-rank condition and  $\|X\|_1$ is
the well-known convex surrogate for $\|\X\|_0$ in
\eqref{eq:firstcond}. We refer to the approach
as \textit{quadratic basis pursuit} (QBP).

One can further consider a noisy counterpart of the QBP problem, where 
some deviation between the measurements and the estimates is allowed.
More specifically, we propose the following \textit{quadratic basis pursuit denoising} (QBPD) problem:
\begin{equation}\label{eq:noisy-SDP}
\begin{array}{rl}
\min_{\X} &\Tr(\X)+\lambda \|\X\|_1 \\
\subjto \quad &\sum_i^N  \|y_i -\Tr ( \PPhi_i \X ) \|^2 \leq \varepsilon,\\
&\X_{1,1}=1,\, \X \succeq 0,
\end{array}
\end{equation}
for some $\varepsilon > 0$.
 
\subsection{Theoretical Analysis}
In this section, we highlight some theoretical results derived for QBP. 
The analysis follows that of CS, and is inspired by derivations given in
\cite{Candes:11,Candes:06,Chai:10,Donoho:06,Candes_2008,berinde:08,bruckstein:09}. 
For further analysis on special cases of QBP and its noisy counterpart
QBPD, please refer to \cite{ohlsson:11m}.

First, it is convenient to introduce a linear operator $\B$:
\begin{equation}
\B: \X\in \Ce^{n\times n} \mapsto \{\Tr (\PPhi_i \X) \}_{1\le i \le N}\in\Ce^{N}.
\label{eq:definition-B}
\end{equation}
 We consider a generalization of the \emph{restricted isometry property} (RIP) of the linear operator $\B$. 
\begin{df}[\bf RIP]\label{def:RIP}
A linear operator $\B(\cdot) $ as defined in \eqref{eq:definition-B} is $(\epsilon,
k)$-RIP 
if 
\begin{equation}\label{eq:RIP}
\left |{ \frac{\| \B(\X) \|^2}{\| \X \|^2}} -1
  \right |<\epsilon
\end{equation} for all $\|\X\|_0 \leq k$ and $\X\neq 0.$
\end{df}
We can now state the following theorem:
\begin{thm}[\bf Recoverability/Uniqueness]\label{thm:rec}
Let $\bar \xx \in \Ce^{n}$ be a solution to \eqref{eq:costraint}. 
If $\X^* \in \Ce^{(n+1)\times (n+1)}$ satisfies
$\yy= \B( \X^*), \, \X^* \succeq 0,\,
\rank(\X^*)=1,\, \X^*_{1,1}=1$ and if $\B(\cdot)$  is a $(\epsilon,2 \| \X^*\|_0)$-RIP linear
operator with  $\epsilon<1$ 
then  $\X^*$ and $\bar \xx$ are unique and  $\X^*_{2:n+1,1}= \bar \xx$.
\end{thm}
\begin{proof}
Assume the contrary \ie $\X^*_{2:n+1,1} \neq \bar \xx$ and hence that $\X^* \neq \begin{bmatrix}1\\ \bar
  \xx \end{bmatrix} \begin{bmatrix} 1 & \bar \xx^\H \end{bmatrix}$.
It is clear that  $ \left \|\begin{bmatrix}1\\ \bar
  \xx \end{bmatrix} \begin{bmatrix} 1 & \bar \xx^\H \end{bmatrix} \right \|_0 \leq \|\X^*\|_0$ and
hence  $\left \|\begin{bmatrix}1\\ \bar
  \xx \end{bmatrix} \begin{bmatrix} 1 & \bar \xx^\H \end{bmatrix} -
\X^*  \right \|_0 \leq 2 \|\X^*\|_0$. Since
$\left \|\begin{bmatrix}1\\ \bar
  \xx \end{bmatrix} \begin{bmatrix} 1 & \bar \xx^\H \end{bmatrix} -
\X^*  \right \|_0 \leq 2 \|\X^*\|_0$, we can apply the RIP
inequality \eqref{eq:RIP} on $\begin{bmatrix}1\\ \bar
  \xx \end{bmatrix} \begin{bmatrix} 1 & \bar \xx^\H \end{bmatrix}  -
\X^*$. If we use that  $\yy = \B(\X^*) = \B \left  (\begin{bmatrix}1\\ \bar
  \xx \end{bmatrix} \begin{bmatrix} 1 & \bar \xx^\H \end{bmatrix}
\right )$ and hence $\B \left  (\begin{bmatrix}1\\ \bar
  \xx \end{bmatrix} \begin{bmatrix} 1 & \bar \xx^\H \end{bmatrix}
-\X^*\right )=0$, we are led to the contradiction
$1<\epsilon$. We therefore conclude that $\X^* = \begin{bmatrix}1\\ \bar
  \xx \end{bmatrix} \begin{bmatrix} 1 & \bar \xx^\H \end{bmatrix}$,
$\X^*_{2:n+1,1}= \bar \xx$ and that  $\X^*$ and $\bar \xx$ are unique.
\end{proof} 
We can also give a bound on the sparsity of $\bar \xx$:
\begin{thm}[\bf Bound on {\small $\left \| \bar
  \xx \right \|_0$} from above]\label{thm:Phaseliftrel1}
Let $\bar \xx$ be the sparsest solution to \eqref{eq:costraint}  and let $\tilde \X $ be
the solution of  QBP \eqref{eq:noiseless-SDP}.  If $\tilde \X$ has
rank 1 then  $\| \tilde \X_{2:n+1,1}\|_0\geq \| \bar
  \xx \|_0$.
\end{thm}
\begin{proof}
Let  $\tilde \X $ be a rank-1  solution of  QBP  \eqref{eq:noiseless-SDP}.
By contradiction, assume $\| \tilde \X_{2:n+1,1}\|_0  < \| \bar \xx \|_0$. Since $\tilde \X_{2:n+1,1}$ satisfies the constraints of
\eqref{eq:costraint}, it is a feasible solution of
\eqref{eq:costraint}. As assumed,  $\tilde \X_{2:n+1,1}$ also
gives a lower objective value than $\bar \xx $ in
  \eqref{eq:costraint}. This is a contradiction since $\bar \xx $ was assumed to be
 the solution of  \eqref{eq:costraint}. Hence we must have that $\|\tilde
 \X_{2:n+1,1}\|_0 \geq \|
 \bar \xx  \|_0$.
\end{proof}
The following result now holds trivially:
\begin{cor}[\bf Guaranteed recovery using RIP]\label{thm:guartee1}
Let $\bar \xx$ be the sparsest solution to \eqref{eq:costraint}. 
 The solution of QBP $\tilde \X$ is equal to  $\begin{bmatrix}1\\ \bar
  \xx \end{bmatrix} \begin{bmatrix} 1 & \bar \xx^\H \end{bmatrix} $ if  it
has rank 1 and
$\B(\cdot) $ is ($\epsilon, 2\|\tilde \X\|_0$)-RIP with $\epsilon<1$.
\end{cor}
\begin{proof}
This follows trivially from Theorem \ref{thm:rec} by realizing that
$\tilde \X$ satisfy all properties of $\X^*$.
\end{proof}
 Given the RIP analysis, it may be  that the linear operator
 $\B(\cdot)$ does satisfy the RIP property
 defined in Definition \ref{def:RIP} with a small enough $\epsilon$, as pointed out in \cite{Candes:11}.
 In these cases, RIP-1  may be considered:
 \begin{df}[\bf RIP-1]\label{def:RIP1}
 A linear operator $\B(\cdot)$  is $(\epsilon,
 k)$-RIP-1 if 
 \begin{equation}
  \left| { \frac{\| \B(\X)\|_1}{\| \X\|_1}}
 -1\right|<\epsilon
 \end{equation}
 for all matrices $\X \neq 0$ and $\|\X\|_0 \leq k$.
 \end{df} 
Theorems \ref{thm:rec}--\ref{thm:Phaseliftrel1} and
Corollary~\ref{thm:guartee1}  all hold
with RIP replaced by
RIP-1 and will not be restated in detail
here. Instead, we summarize the most  important property in the following theorem:  
\begin{thm}[\bf Upper bound and recoverability using RIP-1]\label{thm:bound}
Let $\bar \xx$ be the sparsest solution to \eqref{eq:costraint}. 
 The solution of QBP \eqref{eq:noiseless-SDP}, $\tilde \X$, is equal to  $\begin{bmatrix}1\\ \bar
  \xx \end{bmatrix} \begin{bmatrix} 1 & \bar \xx^\H \end{bmatrix} $ if  it
has rank 1 and
$B(\cdot) $ is ($\epsilon, 2\|\tilde \X\|_0$)-RIP-1 with $\epsilon<1$.
\end{thm}
\begin{proof}
The proof follows trivially from the proof of  Theorem \ref{thm:rec}.
\end{proof} 
The RIP-type argument may be difficult to check for a given matrix and
are more useful for claiming results for classes of matrices/linear
operators. For instance, it has  been shown that random Gaussian matrices satisfy the RIP  with
high probability. However, given realization of a random Gaussian matrix,
it is indeed difficult to check if it actually satisfies the RIP.  
Two alternative arguments are the \emph{spark condition} \cite{Chen:98} and the \emph{mutual coherence} \cite{Donoho:03b,CANDES:2009}.
The spark condition usually gives tighter bounds but is known to be difficult to
compute as well. On the other hand, mutual coherence may give less tight bounds,
but is more tractable. We will focus on mutual coherence, which is defined~as: 
\begin{df}[\bf Mutual coherence]
 For a matrix $\AA$,  define  the \textit{mutual coherence}  as
 \begin{equation}
 \mu(\AA)=\max_{1\leq i,j \leq n, i\neq j} {  \frac{|\AA_{:,i}^\H
  \AA_{:,j} |}{\|\AA_{:,i}\|\|\AA_{:,j}\|}}.
  \end{equation}
\end{df}
Let $\BB$ be the matrix satisfying $\yy= \BB \X^s=\B(\X)$ with $\X^s$ being the vectorized version of $\X$. We are now ready to state the following theorem:
\begin{thm}[\bf Recovery using mutual coherence]\label{thm:guartee2}
Let $\bar \xx$ be the sparsest solution to \eqref{eq:costraint}. 
 The
solution of QBP  \eqref{eq:noiseless-SDP}, $\tilde \X$, is equal to  $\begin{bmatrix}1\\ \bar
  \xx \end{bmatrix} \begin{bmatrix} 1 & \bar \xx^\H  \end{bmatrix} $ if  it
has rank 1 and
$\|\tilde \X\|_0 < 0.5(1+1/\mu(\BB )).$
\end{thm}
\begin{proof}
It follows from
\cite{Donoho:03b} \cite[Thm.~5]{bruckstein:09}
that if
\begin{equation}\label{eq:conduniquenessco}
\| \tilde \X\|_0 < \frac{1}{2} \left ( 1 + \frac{1}{\mu(\BB )}\right)
\end{equation}
then $\tilde \X$ is the sparsest solution to $\yy=\B(\X)$. Since $\begin{bmatrix}1\\ \bar
  \xx \end{bmatrix} \begin{bmatrix} 1 & \bar \xx^\H  \end{bmatrix} $ is by definition the
sparsest rank 1 solution to $\yy=\B(\X)$, it follows that $\tilde \X= \begin{bmatrix}1\\ \bar
  \xx \end{bmatrix} \begin{bmatrix} 1 & \bar \xx^\H  \end{bmatrix}.$  
\end{proof}

\section{Numerical Algorithms}
\label{sec:algorithms}

In addition to the above analysis of guaranteed recovery properties, a
critical issue for practitioners is the efficiency of numerical
solvers that can handle moderate-sized SDP problems. 
Several numerical solvers used in CS may 
be applied to solve nonsmooth SDPs, which include interior-point methods, \eg used in CVX \cite{CVX1}, gradient projection methods
\cite{BertsekasD1999}, and augmented Lagrangian methods (ALM)
\cite{BertsekasD1999}. However, interior-point methods are known to
scale badly to moderate-sized convex problems in general. Gradient
projection methods also fail to meaningfully accelerate  QBP  due to the complexity of the projection
operator. Alternatively, nonsmooth SDPs can be solved by ALM. However, the
augmented primal and dual objective functions are still SDPs,
which are equally expensive to solve in each iteration. 
There also exist a family of iterative approaches, often referred to as \textit{outer
approximation methods}, that  successively approximate
the solution of an SDP by solving a sequence of linear programs  (see \cite{Konno2002}).
These methods approximate the positive semidefinite cone by a set of
linear constraints and refine the approximation in each iteration by
adding a new set of linear constraints. However, we have experienced slow
convergence using these type of methods. In summary, QBP as a nonsmooth SDP is categorically more expensive to solve compared to the
linear programs underlying CS, and the task exceeds the capability of many popular sparse optimization techniques.

In this paper, we propose a novel solver to the nonsmooth SDP underlying QBP via the
\emph{alternating directions method of multipliers} (ADMM, see for
instance \cite{boyd:11} and \cite[Sec. 3.4]{BertsekasParallel})
technique. The motivation to use ADMM is  two-fold:  
\begin{enumerate} 
\item It
scales well to large data sets. \item It is known for its fast
convergence. \end{enumerate}
 There are also a number of strong convergence results
which further motivates the choice \cite{boyd:11}.

To set the stage for ADMM, let $n$ denote the dimension of $\xx$, and let $N$ denote the number of measurements. Then, rewrite \eqref{eq:noiseless-SDP} to the
equivalent SDP
\begin{equation} \label{eq:ADMM}
\begin{array}{rl}
\min_{\X_1,\X_2,\Z} & f_1(\X_1) + f_2(\X_2) + g(\Z), \\
\subjto &\quad \X_1 - \Z = 0, \quad
 \X_2 - \Z = 0,
\end{array}
\end{equation}
where $\X_1 = \X_1^\H \in \Ce^{(n+1) \times (n+1)}$, $\X_2 = \X_2^\H \in \Ce^{(n+1) \times (n+1)}$, $\Z = \Z^\H \in \Ce^{(n+1) \times (n+1)}$, and
\begin{align*}
f_1(\X)   \triangleq & 
\begin{cases}
\Tr(\X) & \mbox{if } y_i =\Tr(\PPhi_i \X),\: i = 1,\dots,N 
\\ & \quad \mbox{and } \X_{1,1} = 1 \\
\infty &  \mbox{otherwise} 
\end{cases}\\
f_2(\X)  \triangleq &
\begin{cases}
0 & \mbox{if } \X \succeq 0 \\
\infty & \mbox{otherwise} 
\end{cases} \\ g(\Z) \triangleq  &\lambda \|\Z\|_1.
\end{align*}
%
%
%

Define two matrices $\Y_1$ and $\Y_2$ as the Lagrange multipliers of the two equality constraints in \eqref{eq:ADMM}, respectively. 
Then the update rules of ADMM lead to the following: 
\begin{equation}\label{eq:ADMMiter}
\begin{array}{rcl}
\X_i^{l+1} & = & \arg\min_{\X = \X^\H} f_i(\X) 
 + \Tr(\Y_i^{l} (\X - \Z^l) ) \\ &+&  \frac{\rho}{2} \|\X - \Z^l\|^2,  \\
\Z^{l+1} & = & \arg\min_{\Z = \Z^\H} g(\Z)  +  \sum_{i = 1}^2 \Tr(\Y_i^{l} \Z) \\ &+&  \frac{\rho}{2}\|\X_i^{l+1} - \Z\|^2, \\
\Y_i^{l+1} & = & \Y_i^l + \rho (\X_i^{l+1} - \Z^{l+1}),
\end{array}
\end{equation}
for $i = 1,2$, where $\rho \geq 0$ is a parameter that enforces consensus between $\X_1$, $\X_2$, and $\Z$.
Each of these steps has a tractable calculation.
After some simple manipulations, we have:
\begin{equation}
\begin{array}{rl}
\X_1^{l+1} = \argmin_{\X = \X^\H}&  \| \X - (\Z^l - \frac{\I + \Y_1^{l}}{\rho}) \|,
\\ \subjto &  \quad y_i = \Tr(\PPhi_i \X),\quad i = 1, \dots, N,
\\ & \quad \X_{1,1} = 1.
\end{array} \label{eq:X1}
\end{equation}

Let $\Bt: \Ce^{(n+1) \times (n+1)} \rightarrow \Ce^{N+1}$ be the augmented
linear operator such that $\Bt(\X) = \begin{bmatrix}\B(\X)\\
  \X_{1,1}\end{bmatrix}$, where $B$ is the linear operator defined by  \eqref{eq:definition-B}.
Assuming that a feasible solution exists, and defining $\Pi_{\Bt}$ as
the orthogonal
projection onto the convex set given by the linear constraints, \ie
$\begin{bmatrix}\yy \\ 1\end{bmatrix} = \Bt(\X)$, the
solution is: $\X_1^{l+1}  = \Pi_{\Bt} (\Z^l - { \frac{\I + \Y_1^{l}}{\rho}}).$
This matrix-valued problem can be solved by converting the linear
constraint on Hermitian matrices into an equivalent constraint on
real-valued vectors.

Next, 
\begin{equation}
\X_2^{l+1}  =  \argmin_{\X \succeq 0}  \left\| \X - \left(\Z^l - {
      \frac{\Y_2^{l}}{\rho}}\right) \right\| = \Pi_{PSD} \left (\Z^l -
  { \frac{\Y_2^{l}}{\rho}} \right),
\label{eq:X2}
\end{equation}
where $\Pi_{PSD}$ denotes the orthogonal projection onto the
positive-semidefinite cone, which can easily be obtained via
eigenvalue decomposition.

Finally, let $\overline \X^{l+1} =  \frac{1}{2} \sum_{i = 1}^2 \X_i^{l+1}$ and similarly $\overline \Y^l$. Then, the $\Z$ update rule can be written:
\begin{equation}
\begin{array}{rl}
\Z^{l+1}  =&  \argmin_{\Z = \Z^\T} \lambda \|\Z\|_1 + \rho \| \Z -
(\overline \X^{l+1} + \frac {\overline \Y^l} {\rho}) \|^2  \\ =&  \soft(\overline \X^{l+1} + \frac {\overline \Y^l} {\rho}, {\frac{\lambda}{2\rho}})
\end{array}
\end{equation}
where $\soft(\cdot)$ in the complex domain is defined with respect to a positive real scalar $q$ as:
\begin{equation}
\label{eq:softthres}
\soft(x,q) = \begin{cases}
0 & \mbox{if } |x| \leq q, \\
\frac{|x| - q}{|x|}x & \mbox{otherwise}.
\end{cases}
\end{equation}
Note that if the first argument is a complex value, the $\soft$ operator
is defined in terms of the magnitude rather than the sign and if it is
a matrix, the  the $\soft$ operator acts element-wise. 

Setting $l = 1, \X_1^l = \X_2^l = \Z^l = I$, where $I$ denotes the identity matrix, and $\rho^l = 1$, setting $l=0$, the Hermitian matrices  $\X_i^{l+1},\Z_i^{l+1},\Y_i^{l+1}$ can now
be iteratively computed using the ADMM iterations
\eqref{eq:ADMMiter}. The stopping criterion of the algorithm is given by:
\begin{align}
\|r^l\|  \leq &  n \epsilon^{abs} + \epsilon^{rel} \max(\| \overline{\X}^l \|, \|\Z^l\|), \\
\|s^l\|  \leq & n \epsilon^{abs} + \epsilon^{rel} \| \overline{\Y}^l \|,
\end{align}
where $\epsilon^{abs}, \epsilon^{rel}$ are algorithm parameters set to $10^{-3}$ and $r^l$ and $s^l$ are the primal and dual residuals, respectively, as:
\begin{align}
r^l  =  &\begin{bmatrix} \X_1^l - \Z^l & \X_2^l - \Z^l\end{bmatrix},
\\ s^l  =&  -\rho \begin{bmatrix} \Z^l - \Z^{l-1} & \Z^l - \Z^{l-1}\end{bmatrix}.\end{align}

We also update $\rho$ according to the rule discussed in \cite{boyd:11}:
\begin{equation}
\begin{array}{rcl}
\rho^{l+1} & = &
\begin{cases}
\tau_{incr} \rho^l & \mbox{if } \|r^l\| > \mu \|s^l\|, \\
\rho^l / \tau_{decr} & \mbox{if } \|s^l\| > \mu \|r^l\|, \\
\rho^l & \mbox{otherwise},
\end{cases}
\end{array}
\end{equation}
where $\tau_{incr}$, $\tau_{decr}$, and $\mu$ are algorithm
parameters. Values commonly used are $\mu = 10$ and $\tau_{incr} =
\tau_{decr} = 2$.

In terms of the computational complexity of the ADMM algorithm, its inner loop calculates the updates of $\X_i$, $\Z$, and $\Y_i$, $i=1,2$. It is easy to see that its complexity is dominated by \eqref{eq:X1} and \eqref{eq:X2}, which is bounded by $\bigO(n^2 N^2 + n^3)$, while the calculation of $\Z$ and $\Y_i$ is linear with respect to the number of their elements.

\section{Experiments}
\label{sec:experiments}
In this section, we provide comprehensive experiments to validate the
efficacy of the QBP algorithms in solving several representative
nonlinear CS which depends quadratically on the unknown. We compare their performance primarily with two existing
algorithms. As we mentioned in Section \ref{sec:introduction}, if an
underdetermined nonlinear system is approximated up to the first
order, the classical sparse solver in CS is basis pursuit. In NLCS
literature, several greedy algorithms have been proposed for nonlinear
systems. In this section, we choose to compare with the
\textit{iterative hard thresholding} (IHT) algorithm in
\cite{Beck:2012} in Section \ref{sec:1Dsimulation} and another greedy
algorithm demonstrated in \cite{Szameit:12} in Section
\ref{sec:subwavelength}.\footnote{Besides the comparisons shown here,
  we have also compared to a number of CPR algorithms \cite{MoravecM2007,Shechtman13}. Not surprisingly,
they performed badly on the general quadratic problems since they
do not account for the linear term.}

 \subsection{Nonlinear Compressive Sensing in Real Domain}
 \label{sec:1Dsimulation}
 In this experiment, we illustrate the concept of nonlinear compressive
 sensing. Assume that there is a cost associated with sampling and
 that we would like to recover  $\zz_0 \in \Re^{m}$, related to our
 samples $y_i \in \Re,\,i=1,\dots,N,$ via
\begin{equation}\label{eq:nonlinearexp}
y_i=f_i(\zz_0),\quad  i=1,\dots,N,
\end{equation}
using as few samples as possible.
Also, assume that there is a sparsifying basis ${\bf D} \in \Re^{m \times n}$, possibly overcomplete,
such that  
\begin{equation}
\zz_0 =  {\bf D} \xx_0,\quad  \text{with } \xx_0 \text{ sparse.}  
\end{equation}
Hence, we have
\begin{equation}\label{eq:nonlinearexp2}
y_i=f_i( {\bf D} \xx_0 ),\quad  i=1,\dots,N,
\end{equation}
with $\xx_0$ a sparse vector.
If we  approximate the nonlinear equation system
\eqref{eq:nonlinearexp2} using a second order Maclaurin expansion we
endup with a set of quadratic equations, 
\begin{equation}\label{eq:nonlinearexp3}
y_i=f_i(0)  + \nabla f_i( 0)  {\bf D}   \xx_0+ \xx_0^\T{\bf D}^\T \frac{\nabla^2 f_i(0)}{2}  {\bf D} \xx_0,\;  i=1,\dots,N.
\end{equation}
Hence, we can use QBP to recover $\xx_0$ given $\{ f_i(\xx), y_i \}_{i=1}^N$
and $\bf D$. 

In particular, let ${\bf D}=\I$, $n=m=20$, $N=25$,  $f_i(\xx)=a_i+\bb_i^\T \xx+\xx^\T \Q_i\xx,\, i=1,\dots,N$,
 and generate $\{y_i\}_{i=1}^{N}$ by sampling $\{a_i,\bb_i,\Q_i
 \}_{i=1}^N$ from a unitary Gaussian distribution. Let $\xx_0$ be a
 binary vector with three elements different than zero.
Given $\{y_i, a_i,\bb_i,\Q_i \}_{i=1}^N$, the task is now to recover $\xx_0$. The results of this simulation
are shown in Figure \ref{fig:1D}.

\begin{figure}[h!]
 \centering
\subfigure[Ground truth.]{\includegraphics[width=0.49\columnwidth]{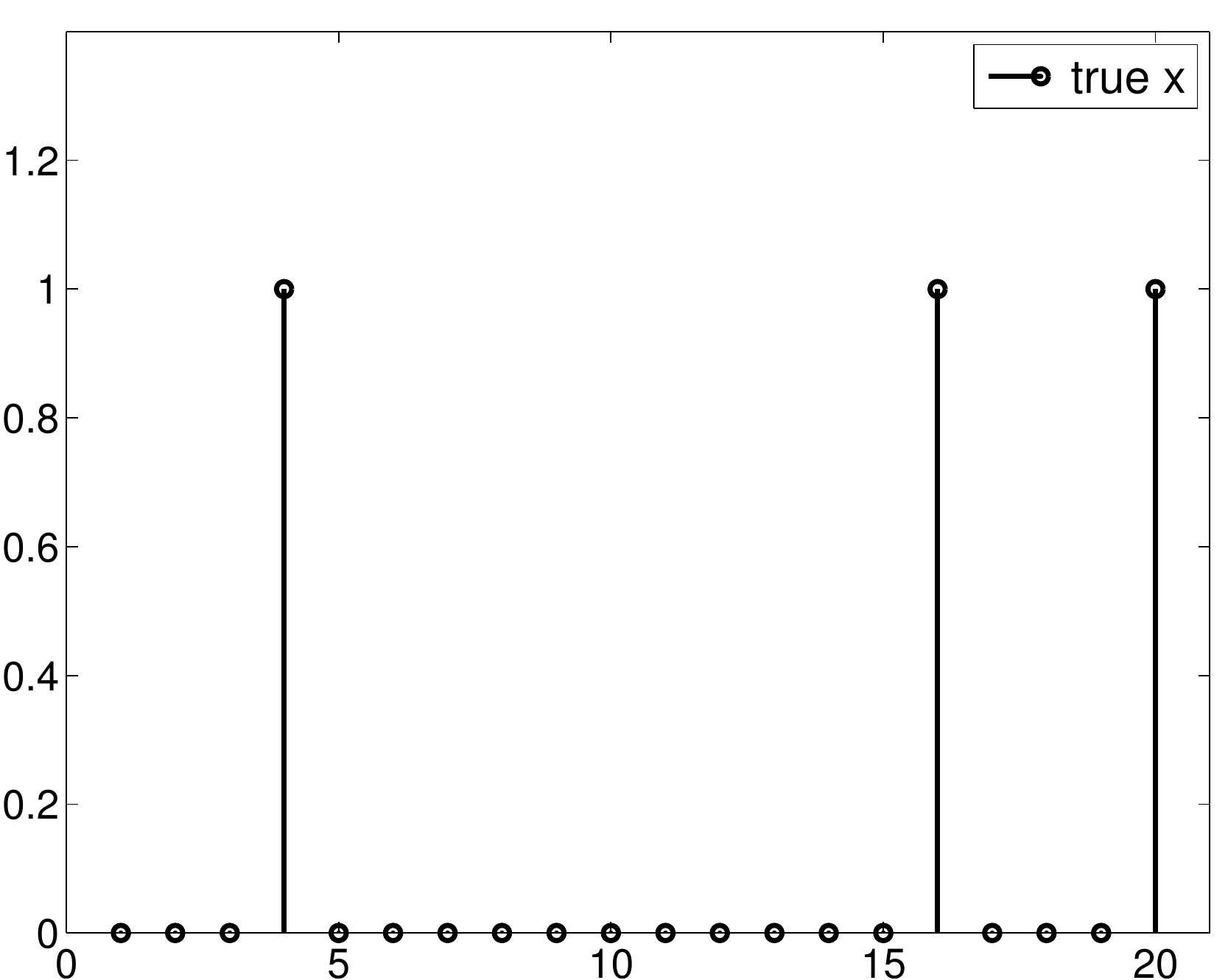}}\\
\subfigure[QBP with $\lambda = 50$.]{\includegraphics[width=0.49\columnwidth]{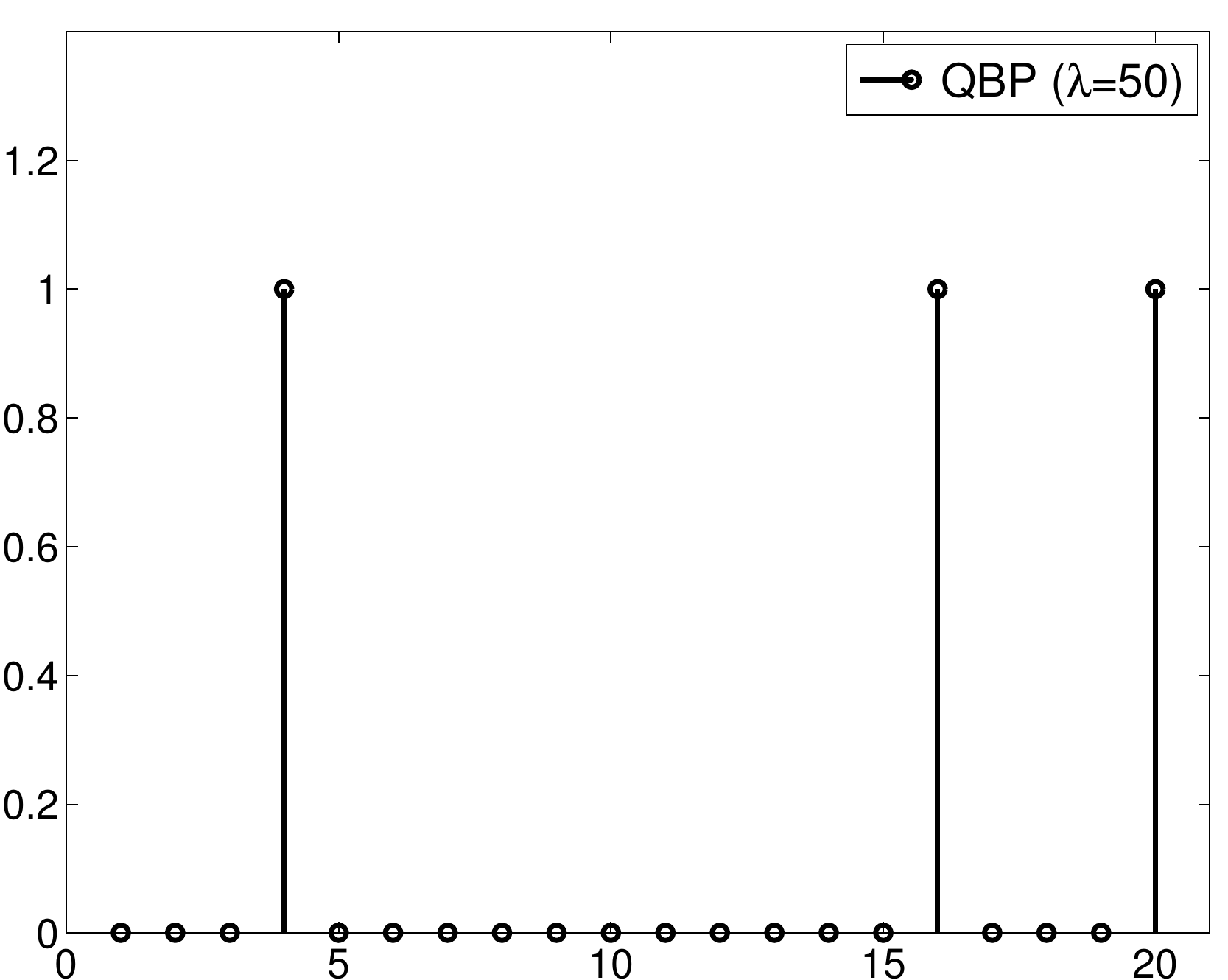}}
\subfigure[QBP with $\lambda = 0$.]{\includegraphics[width=0.49\columnwidth]{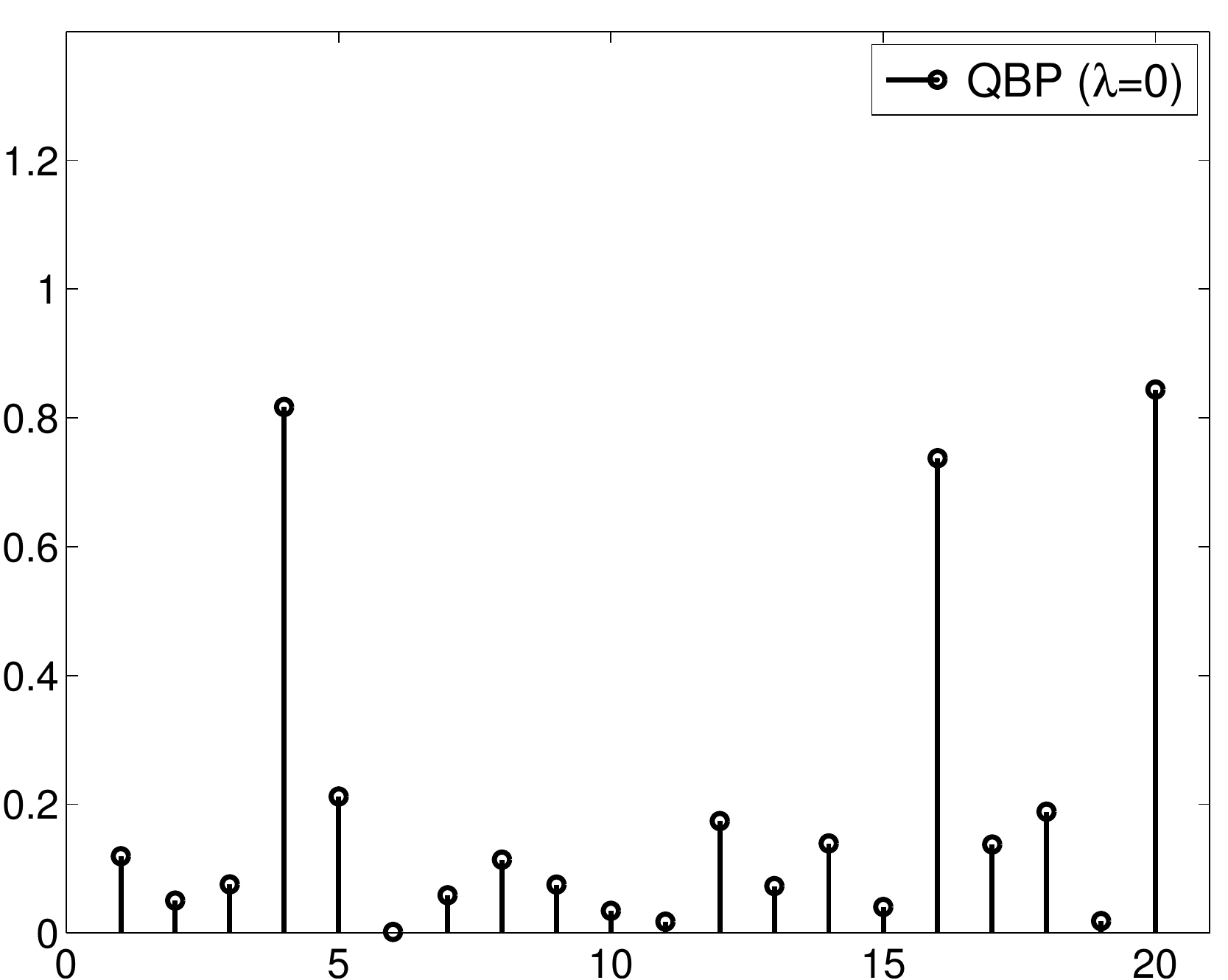}}\\
\subfigure[Basis pursuit.]{\includegraphics[width=0.49\columnwidth]{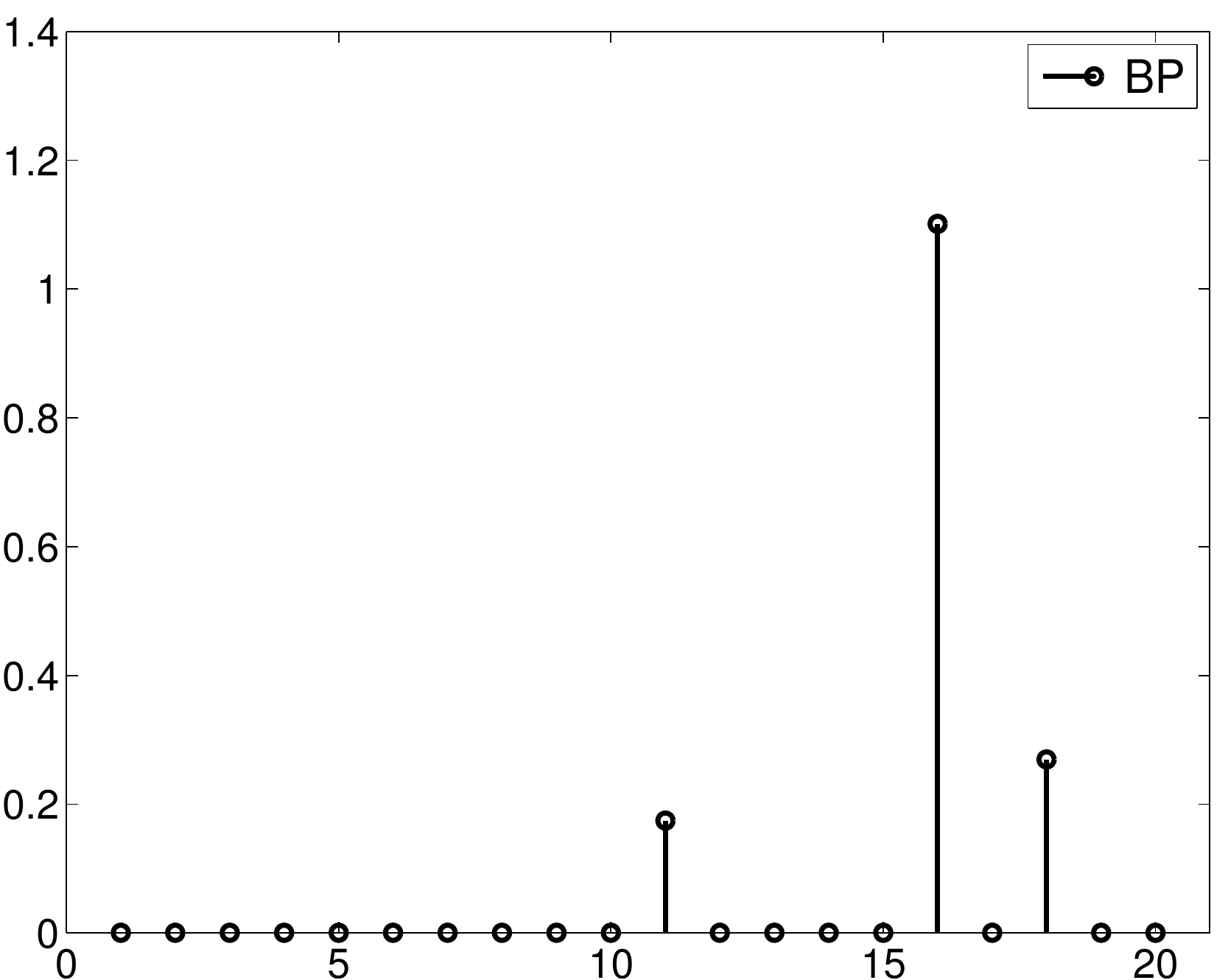}}
\subfigure[Iterative hard thresholding.]{\includegraphics[width=0.49\columnwidth]{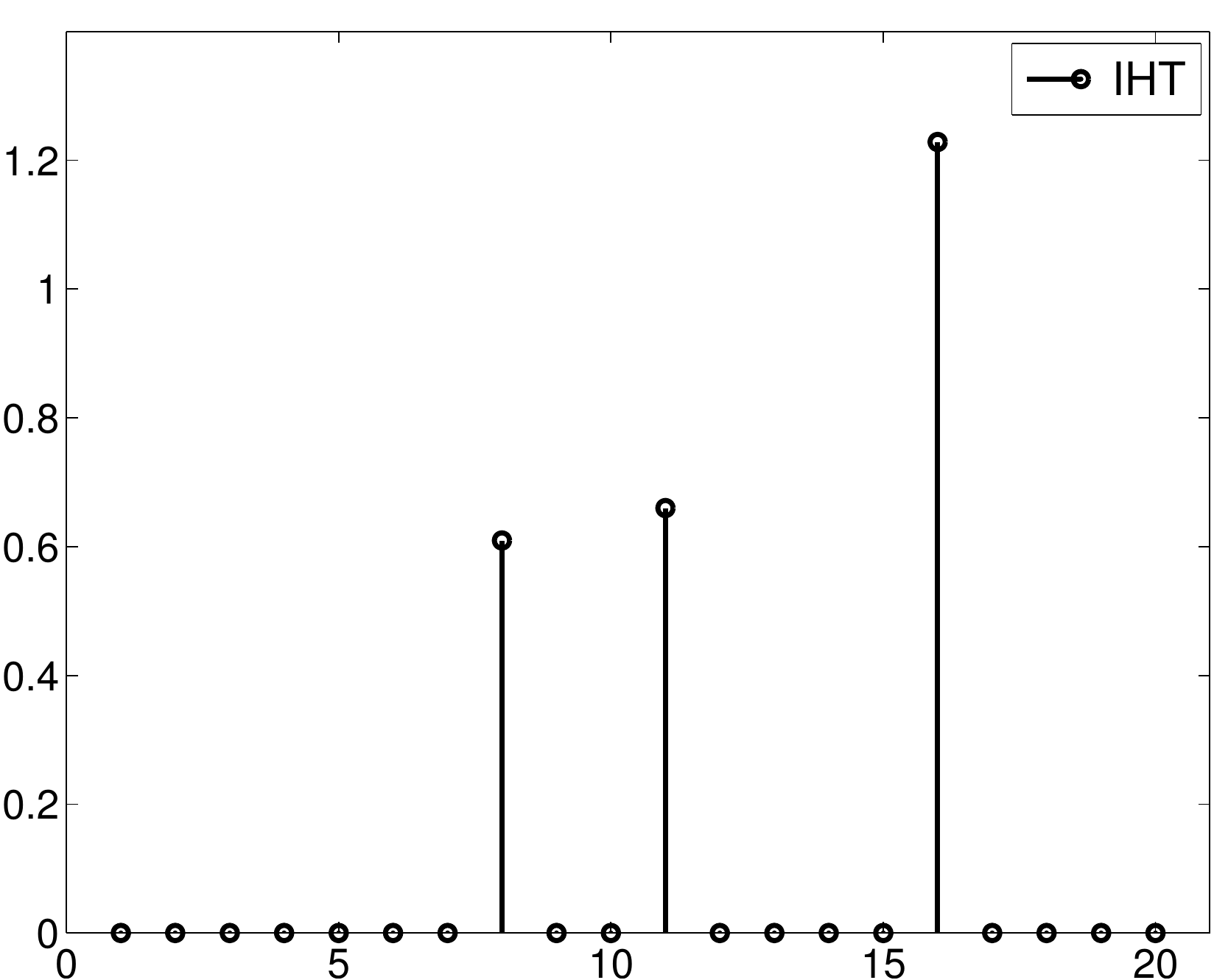}}
 \caption{Estimated 20-D sparse signals measured in a simulated quadratic system of equations.
The QBP solution perfectly recovers the ground truth with $\lambda=50$, while the remaining algorithms
fail to recover the correct sparse coefficients.}\label{fig:1D}
 \end{figure} 

First, as the noiseless measurements are generated by a quadratic system of equations, it is not surprising that QBP perfectly
recovers the sparse signal $\xx_0$ when $\lambda=50$. One may wonder whether in the 25-D ambient space, the solution $\xx_0$ is unique.
To show that the solution is not unique, we let
$\lambda=0$ and again apply QBP. As shown in Figure \ref{fig:1D} (c), the solution is dense and it also satisfies the quadratic constraints.
Therefore, we have verified that the system is underdetermined and there exist multiple solutions.

Second, in Figure \ref{fig:1D} (d), we approximate \eqref{eq:nonlinearexp3} only up to the first order and set $\Q_i=0, i=1,\dots,N$. The approximation enables us to employ the classical basis pursuit algorithm in CS to seek the best 3-sparse estimate $\xx$. As expected, the approximation is not accurate enough, and the estimate is far from the ground truth.

Third, we implement the iterative hard thresholding (IHT) algorithm in \cite{Beck:2012}, and the correct number of nonzero coefficients in $\xx_0$ is also provided to the algorithm. Its estimate is given in Figure \ref{fig:1D} (e). As IHT is a greedy algorithm, its performance is affected by the initialization. In Figure \ref{fig:1D} (e), the initial value is set by $\xx=0$, and the estimate is incorrect.

Finally, we note that the advantage of using general CS theory is that fewer samples are needed to recover a source signal from its observations. This remains true for NLCS presented in this paper. However, as \eqref{eq:nonlinearexp} and \eqref{eq:nonlinearexp3} are nonlinear equation systems, typically  $N\gg m$ measurements  are required for recovering a unique solution. In the same simulation shown in Figure \ref{fig:1D}, one could ignore the sparsity constraint (\ie, by letting $\lambda=0$ in Figure \ref{fig:1D} (c)), and it would require $N'=40$ observations for QBP to recover the unique solution, which is exactly the ground-truth signal.

Clearly, Figure \ref{fig:1D} is only able to illustrate one set of simulation results. To more systematically demonstrate the accuracy of the four algorithms in probability, a Monte
Carlo simulation is performed that repeats the above simulation but with different randomly generated $\xx_0$ and $\{a_i,\bb_i,\Q_i\}$.
Table~\ref{tab:first} shows the rates of successful recovery. We can
see QBP achieves the highest success rate, which is followed by
IHT. BP and the dense QBP solution basically fail to return enough
good results. $\lambda=50$ was used in all trials.

\begin{table}[h!]
\caption{The percentage of correctly recovering $\xx_0$ in 100 trials.}
\label{tab:first}
\begin{center}
{ \begin{tabular}{|l|cccc|}
\hline
Method & \text{QBP} ($\lambda=50$)&\text{QBP} ($\lambda=0$) & \text{BP}&\text{IHT} 
\\ \hline 
Success rate &   79\%  & 5\%  &3\%  &54\%\\
\hline
\end{tabular} }
\end{center}
\end{table}


\subsection{The Shepp-Logan Phantom}\label{sec:shepp}
In this experiment, we consider recovery of 
images from random samples. More specifically, we formulate an example
of the CPR problem in the QBP framework using the Shepp-Logan
phantom. Our goal is to show that using the QBPD algorithm  provides approximate solutions that are visually close to the ground-truth images.

Consider the ground-truth image in Figure \ref{fig:shepplogan}. This $30 \times 30$ Shepp-Logan phantom has a 2D
Fourier transform with 100 nonzero complex  coefficients.  We generate $N$
linear combinations of pixels, and then measure the square of the measurements. This relationship can be written as:
\begin{equation}
\yy = |{\bf A}\xx|^2 = \{\xx^\H \aa_i  \aa_i^\H\xx \}_{1\le i \le N},
\end{equation}
where ${\bf A}={\bf R}{\bf F}$ is the concatenation of a random matrix ${\bf R}$ and the Fourier basis ${\bf F}$, and the image ${\bf F}\xx$ is represented as a stacked vector in the 900-D complex domain. The CPR problem minimizes the following objective function:
\begin{equation}
\min_\xx\| \xx\|_1 \quad \subjto \quad \yy = | {\bf A}\xx |^2 \in\Re^N.
\label{eq:CPR}
\end{equation} 

Previously, an SDP solution to the non sparse phase retrieval problem was proposed in \cite{Candes:11}, which is called \textit{PhaseLift}. 
In a sense, PhaseLift can be viewed as a special case of the  QBP solution in \eqref{eq:noiseless-SDP} where $\lambda=0$, namely, the sparsity constraint is not enforced. In Figure \ref{fig:shepplogan} (b), the recovered result using PhaseLift is shown with $N=2400$.

To compare visually the performance of the QBP solution when the sparsity constraint is properly enforced, two recovered results are shown in Figure \ref{fig:shepplogan} (c) and (d) with  $N=2400$ and $1500$, respectively. Note that the number of measurements with respect to the sparsity in $\xx$ is too low for both QBP and PhaseLift to
perfectly recover $\xx$. Therefore, in this case, we employ the noisy version of the algorithm QBPD to recover the image. Wee can clearly see from the illustrations that QBPD provides a much better approximation and outperforms PhaseLift visually even though it uses considerably fewer measurements.
\begin{figure}[h!]
 \centering
 \subfigure[Ground truth]{\includegraphics[width = 0.4\columnwidth]{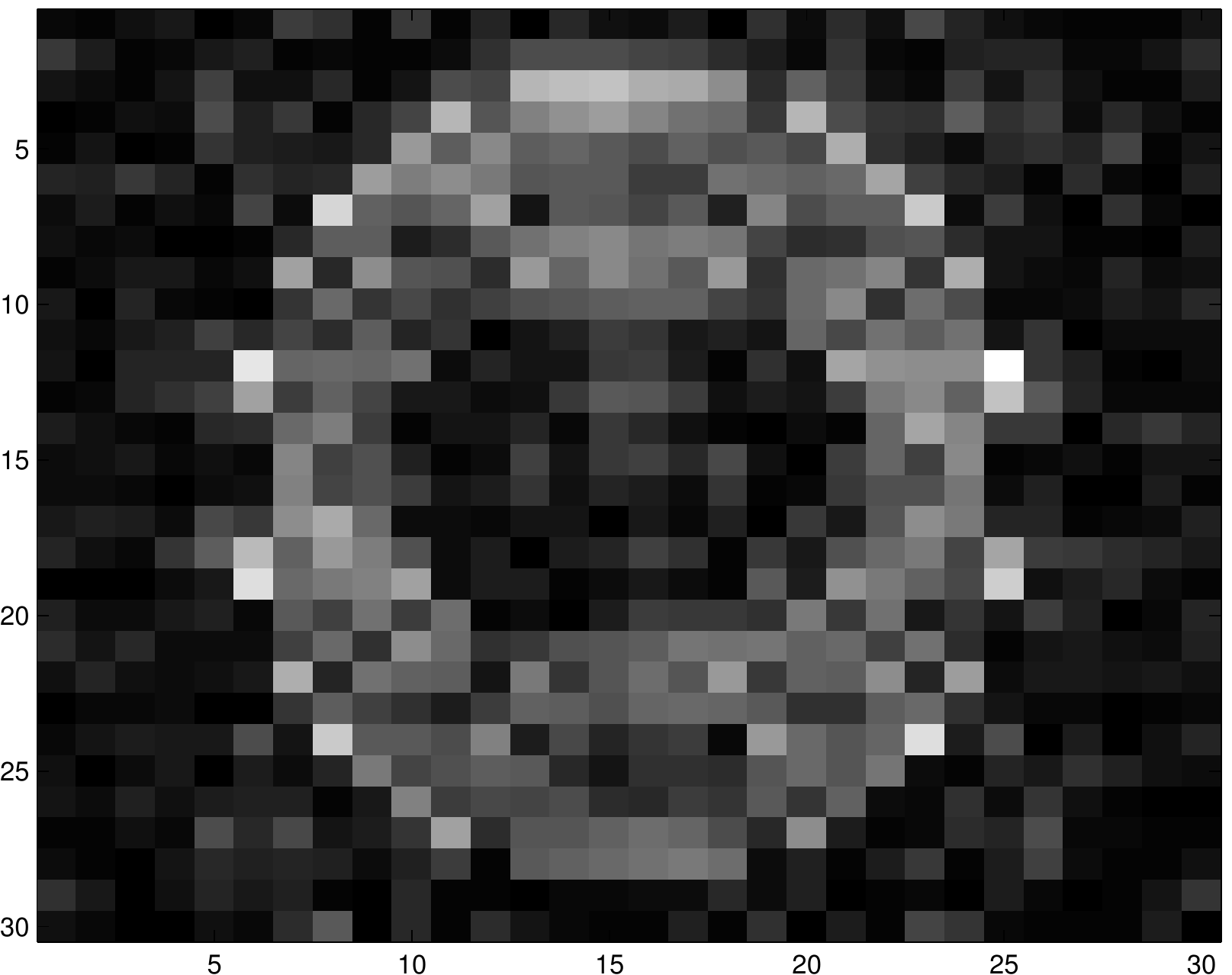}}
\subfigure[PhaseLift with $N=2400$] {\includegraphics[width = 0.4\columnwidth]{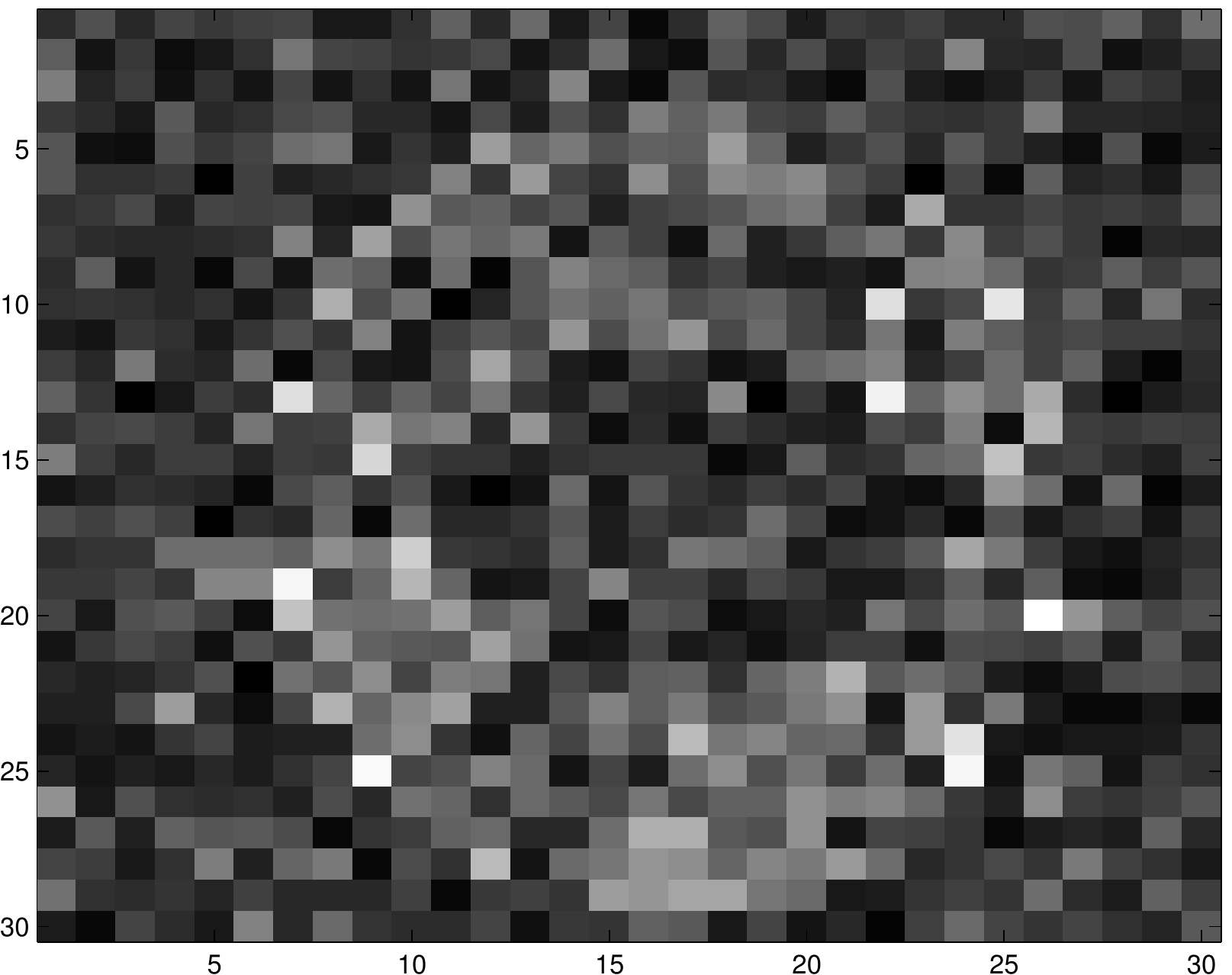}} \\
\subfigure[QBPD with  $N=2400$] {\includegraphics[width =0.4\columnwidth]{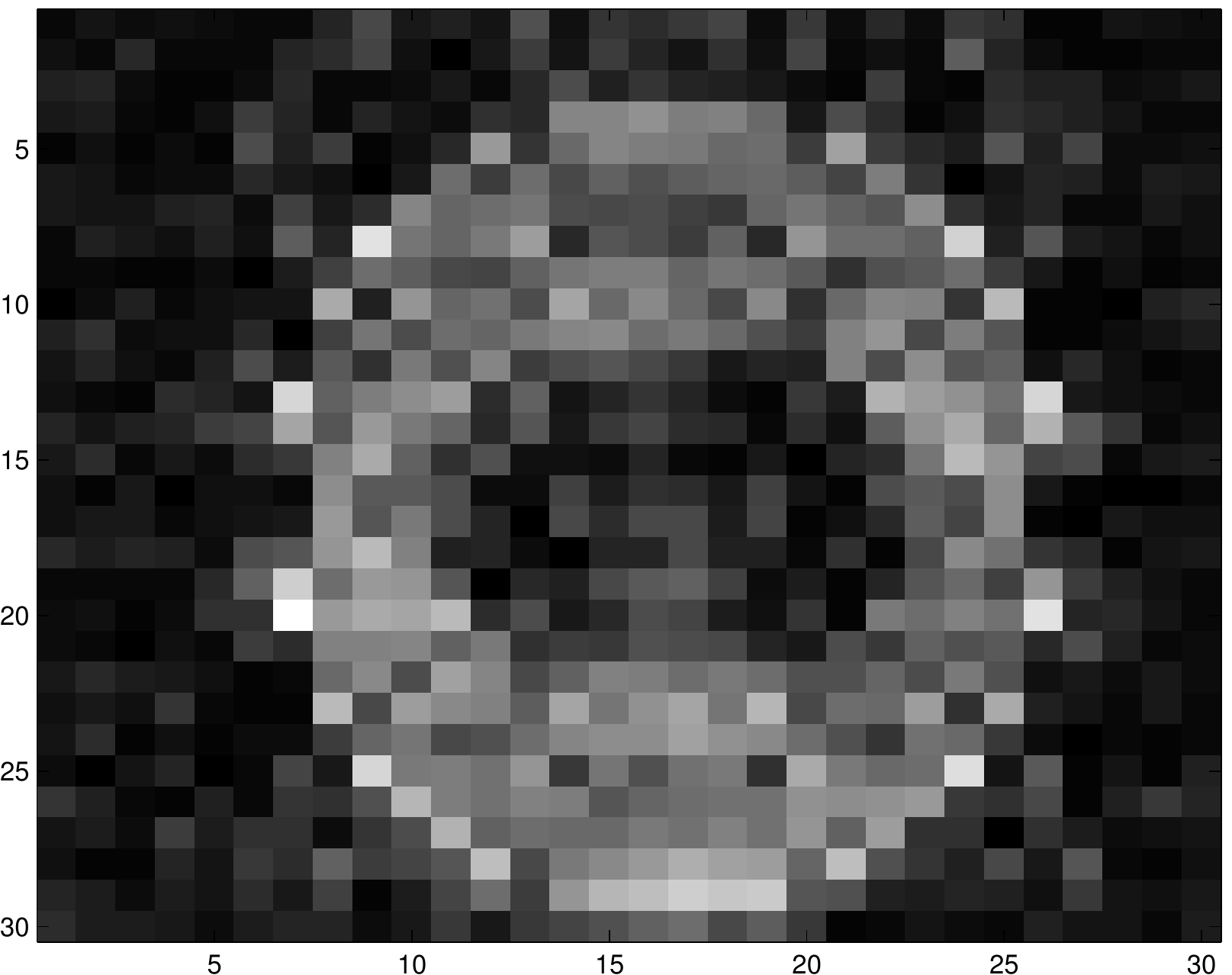}}
\subfigure[QBPD with $N=1500$] {\includegraphics[width = 0.4\columnwidth]{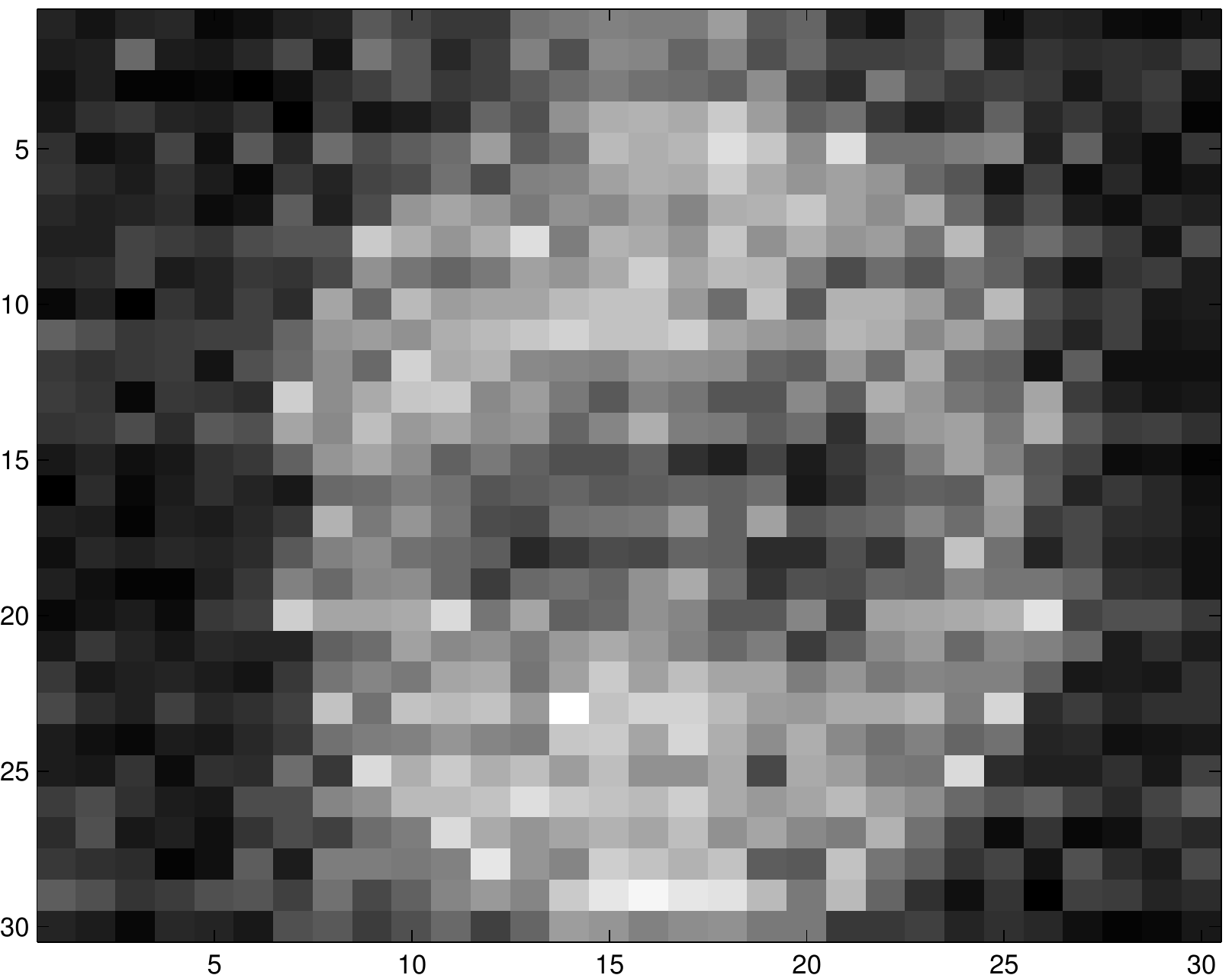}}
 \caption{Recovery of a Shepp-Logan Image by PhaseLift and QBPD.}
 \label{fig:shepplogan}
 \end{figure}

\subsection{Subwavelength Imaging}
\label{sec:subwavelength}
In this example, we present an example in sub-wavelength coherent diffractive imaging. The experiment and the data collection were conducted by
\cite{Szameit:12}. 

Let  $y_i,\,i=1,\dots,N,$ be intensity samples of a
2D diffraction pattern. The diffraction pattern is the result of a 532
nm laser beam passing through an arrangement of holes made on a
opaque piece of glass. The task is to decide the location of the holes out of a
number of possible locations. 

It can be shown that the relation between
the intensity measurements and the arrangements of holes is of the following type:
\begin{equation}
y_i=|\aa_i^\H \xx|^2,\quad i=1,\dots,N,
\end{equation}
where $y_i\in \Re,\,i=1,\dots,N,$ are intensity measurements,
$\aa_i\in \Ce^n,\, i=1,\dots,N,$ are known complex vectors and $\xx
\in \Re^n,$ is the sought entity, each element giving the likelihood
of a hole at a given location. 

We use QBPD
with $\varepsilon=0.0012$ and $\lambda= 100$. 89 measurements were selected  by
taking every 200th intensity measurement from the dataset of
\cite{Szameit:12}. The quantity $\xx$ is from the setup of the
experiment known  to be real and $a_i=\bb_i=\cc_i=0$.  We hence have
\begin{equation}
y_i=\xx^\T \Q_i \xx=|\aa_i^\H \xx|^2,\quad i=1,\dots,N,
\end{equation}
with $\Q_i=\aa_i \aa_i^\H \in \Ce^{n \times n},\,i=1,\dots,N,$ and
$\xx \in \Re^n$. 

The resulting estimate is
given to the left in Figure~\ref{fig:recovered}. 
The result deviates from the ground truth and the result presented in
\cite{Szameit:12}  (shown in Figure~\ref{fig:recovered} right), and it actually finds a more sparse pattern. 
It is interesting to note that both estimates are however within the noise level estimated in \cite{Szameit:12}:
\begin{equation}
\frac{1}{N}\sum_i^N (y_i -|\aa_i^\H \xx|^2 )^2 \leq 1.8 \times 10^{-6}.
\end{equation}
 Therefore, under the same noise assumptions, the two solutions are equally likely to lead to the same observations $\yy$.
 However, knowing that there is a solution within the noise
 level that is indeed sparser than the ground-truth pattern, it should \emph{not} be the optimal solution to have recovered the ground truth, since there exists a sparser solution.
 
 \begin{figure}[h!]
\centering
\includegraphics[width=0.5\columnwidth]{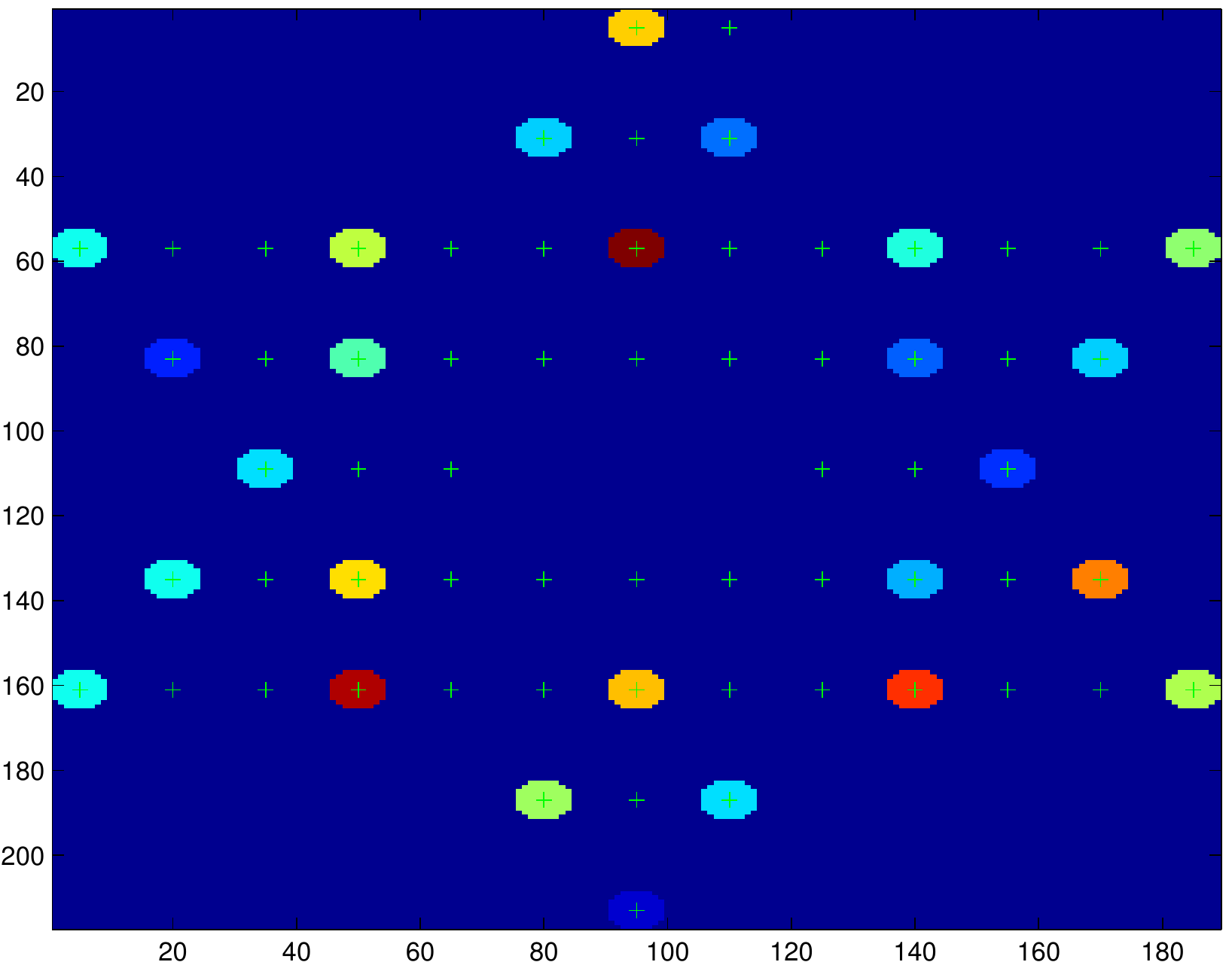}\includegraphics[width=0.5\columnwidth]{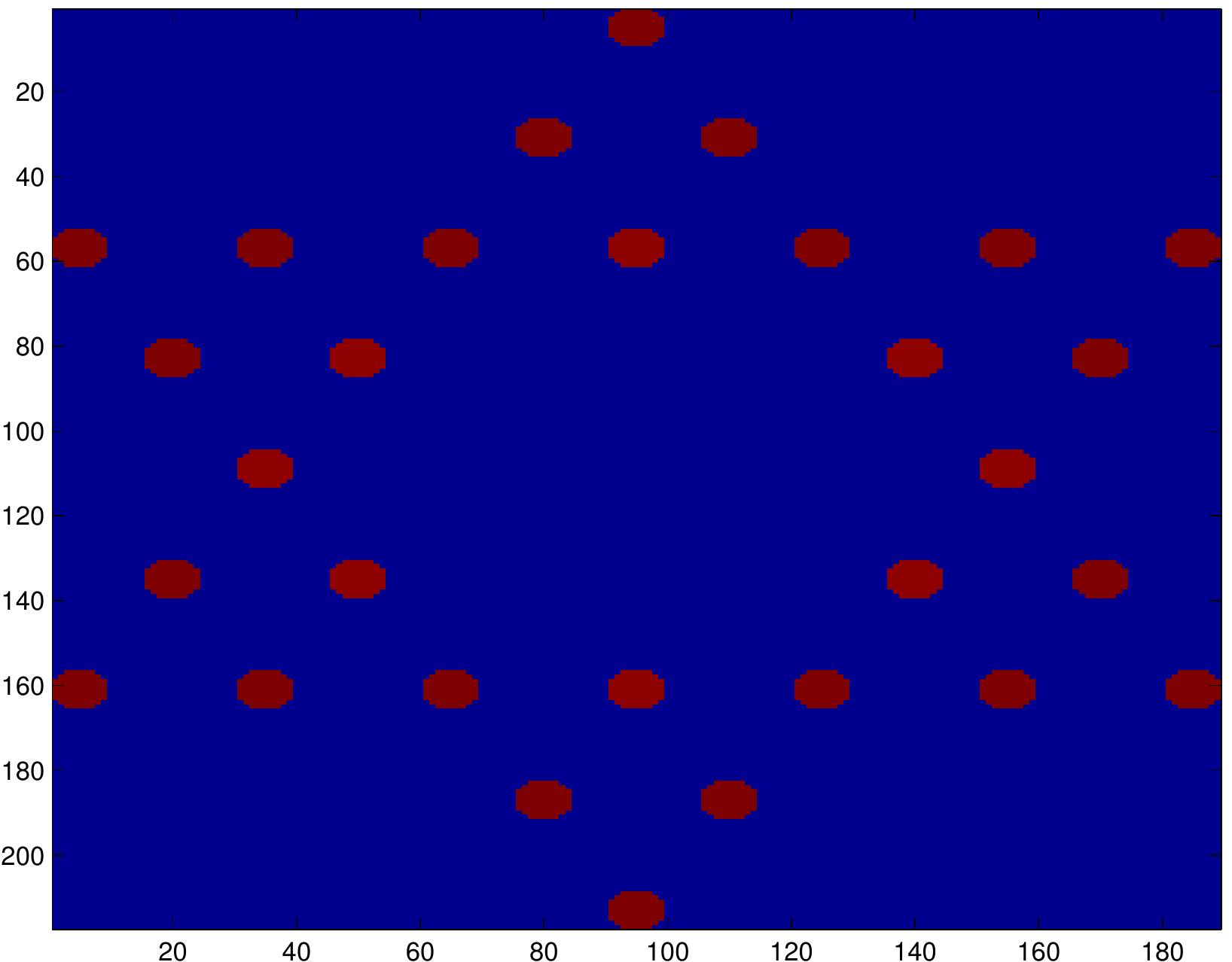}
\caption{The estimated sparse vector $\xx$. The crosses mark possible
  positions for holes, while the dots represent the recovered nonzero
  coefficients. {\bf Left:} Recovered pattern by QBPD.  Note that this
  estimate is sparser than the ground truth but within the estimated
  noise level. {\bf Right:} Recovered pattern by the compressive phase retrieval method used in \cite{Szameit:12}.}\label{fig:recovered}
\end{figure} 

\section{Conclusion}

Classical compressive sensing assumes a linear relation
between samples and the unknowns. The ability to more accurately
characterize nonlinear models has the potential to improve the results
in both existing compressive sensing applications and those where a
linear approximation does not suffice, \eg  phase retrieval.

This paper presents an extension of classical compressive sensing to
quadratic relations or second order Taylor expansions of the
nonlinearity relating measurements and the unknowns. The novel
extension is based on lifting and convex relaxations and the final
formulation takes the form of a  SDP. The proposed method, quadratic
basis pursuit,  inherits
properties of basis pursuit and classical compressive sensing and conditions for perfect
recovery etc are derived. We also give an efficient numerical implementation.

\section*{Acknowledgement}
The authors would like to acknowledge useful discussions and inputs
from Yonina C. Eldar,  Mordechai  Segev,
Laura Waller, Filipe Maia, Stefano Marchesini and Michael Lustig. We also
want to acknowledge the authors of \cite{Szameit:12} for kindly
sharing their data with us.

Ohlsson is partially supported by the Swedish Research
  Council in the Linnaeus center CADICS, the European Research Council
   under the advanced grant LEARN, contract 267381, by  a postdoctoral grant from the Sweden-America
   Foundation, donated by ASEA's Fellowship Fund, and by a postdoctoral
   grant from the Swedish Research Council. Yang is supported in part
   by ARO 63092-MA-II and DARPA FA8650-11-1-7153.
   
\bibliographystyle{IEEEtran}
\bibliography{refHO}
\end{document}